\documentclass[11pt]{article}
\usepackage{odonnell}
\renewcommand{\bits}{\{\pm 1\}}
\renewcommand{\bn}{\bits^n}
\renewcommand{\btb}{\co \bn \to \bits}
\renewcommand{\btR}{\co \bn \to \R}
\newcommand{\semi}{\mathbin{;}}
\newcommand{\MI}[2]{I(#1 \semi #2)}
\newcommand{\stdgauss}{\mathrm{N}(0,1)}

\newtheorem{innercustomlem}{Lemma}
\newenvironment{customlem}[1]
  {\renewcommand\theinnercustomlem{#1}\innercustomlem}
  {\endinnercustomlem}

\begin{document}

\title{Continuous analogues of the Most Informative Function problem}
\author{Guy Kindler\thanks{School of Computer Science and Engineering, Hebrew University. Supported by BSF fund no. 2008477} \and Ryan O'Donnell\thanks{Department of Computer Science, Carnegie Mellon. Supported by NSF grants CCF-0747250 and CCF-1116594.  \indent \, $^{\S}$Some of this work performed while the author was at the Bo\u{g}azi\c{c}i University Computer Engineering Department, supported by Marie Curie International Incoming Fellowship project number 626373, and by BSF fund no. 2008477.}$\ {}^\S$ \and David Witmer${}^\dagger$\thanks{Supported by the NSF Graduate Research Fellowship Program under grant DGE-1252522.}}

\maketitle

\begin{abstract}
    In 2013, Courtade and Kumar posed the following problem: Let $\boldsymbol{x} \sim \{\pm 1\}^n$ be uniformly random, and form $\boldsymbol{y} \sim \{\pm 1\}^n$ by negating each bit of $\boldsymbol{x}$ independently with probability~$\alpha$.  Is it true that the mutual information $I(f(\boldsymbol{x}) \mathbin{;} \boldsymbol{y})$ is maximized among $f:\{\pm 1\}^n \to \{\pm 1\}$ by $f(x) = x_1$?  We do not resolve this problem.  Instead, we resolve the analogous problem in the settings of Gaussian space and the sphere.  Our proof uses rearrangement.
\end{abstract}

\section{The Courtade--Kumar Conjecture}

In 2013, Courtade and Kumar~\cite{KC13,CK14} made the following conjecture:
\begin{named}{The Courtade--Kumar ``Most Informative Boolean Function'' Conjecture} \

Let $\bx \sim \bn$ be uniformly random and form $\by \sim \bn$ by negating each bit of $\bx$ independently with probability~$\alpha$.  Then for any $f \btb$ it holds that $\MI{f(\bx)}{\by} \leq 1-h(\alpha)$.  (This bound is achieved by any $f$ of the form $f(x) = \pm x_i$.)
\end{named}

The conjecture attracted fairly widespread attention; it is currently unresolved (though~\cite{CK14} verifies it for $n \leq 7$).  Courtade offers a prize of \$100 for a proof or disproof~\cite{Cou14}.\\

Let us briefly discuss the notation used in this problem. First, we henceforth assume $\alpha \leq \frac12$, as it's easy to see the problem is unchanged if $\alpha$ is replaced by $1-\alpha$. The mutual information $\MI{\bA}{\bB}$ of two discrete random variables is defined to be $H(\bB) - H(\bB | \bA)$. Here $H(\bB)$ denotes entropy, namely $H(\bB) = \sum_{b} \Pr[B=b] \log (\frac{1}{\Pr[B=b]})$ (with $\log = \log_2$), and $H(\bB | \bA)$ denotes conditional entropy, namely the expected value of $H(\bB \mid \bA = a)$ when $a$ is distributed as $\bA$.  For $\beta \in [0,1]$ we write $h(\beta) = \beta \log(\frac{1}{\beta}) + (1-\beta)\log(\frac{1}{1-\beta})$ for the entropy of the two-valued random variable that is $-1$ with probability~$\beta$ and $+1$ with probability~$1-\beta$.  We will also be using traditional notation from the field of analysis of Boolean functions~\cite{OD14}.  In particular, recall that $(\bx,\by)$ is said to be a pair of \emph{$\rho$-correlated random strings}, where $\rho = \E[\bx_i \by_i] = 1-2\alpha \geq 0$ (and $(\by,\bx)$ has the same distribution).  Also recall that for $f \btR$, the function $\T_\rho f \btR$ is defined by $\T_\rho f(x) = \E[f(\by) \mid \bx = x]$. Note that $\E[\T_\rho f] = \E[f]$ (where we use the shorthand $\E[g] = \E[g(\bx)]$). Using this notation, and defining for convenience
\begin{equation}    \label{eqn:Phi}
    \Phi \co [-1,1] \to [0,1], \qquad \Phi(t) = 1 - h(\half - \half t) = \tfrac{1}{\ln2} \cdot \left(\tfrac{1}{2\cdot 1} \cdot t^2 + \tfrac{1}{4 \cdot 3} \cdot t^4 + \tfrac{1}{6\cdot 5} \cdot t^6 + \cdots \right)
\end{equation}
we have
\begin{align*}
    \MI{f(\bx)}{\by} &= H(f(\bx)) - H(f(\bx) | \by) = h(\half + \half \E[f]) - \E[h(\half + \half \T_\rho f(\bx))] \\
    & = \E[\Phi(\T_\rho f(\bx))] - \Phi(\E[\T_\rho f(\bx)]) = \Ent^{\Phi}[\T_\rho f],
\end{align*}
where in the last equality we are using the $\Phi$-entropy notation from, e.g.,~\cite{Cha04}.  Thus we have the following equivalent formulation:

\begin{named}{Courtade--Kumar Conjecture (equivalently)}  
    For $f \btb$ and $\rho \in [0,1]$ it holds that $\Ent^{\Phi}[\T_\rho f] \leq \Phi(\rho)$, where $\Phi$ is as in~\eqref{eqn:Phi}.
\end{named}

\noindent We remark that $\Phi$ is very close to the function $t \mapsto t^2$, and that the analogous statement
\[
    \Ent^{t \to t^2}[\T_\rho f] = \Var[\T_\rho f] \leq \rho^2 = (1-2\alpha)^2,
\]
(with equality if and only if $f(x) = \pm x_i$, presuming $0 < |\rho| < 1$) has a rather trivial Fourier-theoretic proof.  (Combine~\cite[Prop.~1.13, Prop~2.47, Ex.~1.19(a)]{OD14}.)

\subsection{Prior work}

The Courtade--Kumar Conjecture is a very natural one in information theory and the analysis of Boolean functions.  Courtade and Kumar report that their original motivation came from the work~\cite{KKBS14}, which observed that among $f \btb$ with $\E[f] = \mu \geq 0$, the quantity $\MI{f(\bx)}{\bx_1}$ is maximized by those~$f$ with $f(x) \geq x_1$.  In turn, \cite{KKBS14}~was motivated by a work~\cite{SJ08} on the regulatory network of E.~coli.  A connection between the conjecture and cryptography is discussed in~\cite{CVM+13}.  Finally, Courtade and Kumar also offered a motivation from gambling (stock markets, horse races), and in fact closely related problems were studied earlier by Erkip and Cover~\cite{EC98}.  In~\cite{CK14} the weaker result $\MI{f(\bx)}{\by} \leq (1-2\alpha)^2 = \rho^2$ is attributed to Erkip~\cite{Erk96}.

There are some natural weakenings of the conjecture that are still open.  For example, it is natural to expect that maximizing~$f$ are \emph{unbiased}, meaning $\E[f] =0$.  However, the conjecture remains open even under this assumption.  Courtade and Kumar also left open the weaker conjecture ``$\MI{f(\bx)}{g(\by)} \leq 1-h(\alpha)$ for $f,g \btb$'', but remarked that it is an exercise assuming both~$f$ and~$g$ are unbiased.  Bogdanov and Nair~\cite{BN13} have apparently proved this weaker conjecture under the assumption that $f = g$ (and $\alpha \geq \frac12$); see also~\cite{AGKN13}, in which the weaker conjecture is reduced to an explicit three-dimensional numerical inequality which, empirically, appears to be true. Courtade and Kumar also proved the weakening $\sum_{i=1}^n \MI{f(\bx)}{\by_i} \leq 1 - h(\alpha)$ under the  assumption that $f$~is unbiased.

Certain strengthenings of the Courtade--Kumar Conjecture have also been considered; see, e.g., the information theory work~\cite{CVM14}.  Another interesting example comes from the work of Chandar and Tchamkerten~\cite{CT14}, who considered the more general conjecture 
\begin{equation}            \label{eqn:CT}
    \frac{\MI{f(\bx)}{\by}}{k} \leq 1-h(\alpha) \qquad \text{for all $f \co \bn \to \bits^k$}.
\end{equation}
Chandar and Tchamkerten generalized the Erkip--Cover bound by showing that one can take ${(1-2\alpha)^2}$ on the right-hand side above, for all~$k$. However they also showed that~\eqref{eqn:CT} is too strong; in fact, a right-hand side of $(1-2\alpha)^2$ can be \emph{achieved} in the limit when first $n \to \infty$ and then $k \to \infty$.  In particular, by taking $f$ to be the indicator of certain perfect codes, they showed that~\eqref{eqn:CT} can fail when, e.g., $n = 15$, $k = 11$, $\alpha \in [0.05, 0.5]$.

In recent work, Ordentlich, Shayevitz, and Weinstein \cite{OSW15} showed that the Courtade-Kumar Conjecture holds for unbiased functions when $\alpha$ is very close to $0$ or $\frac{1}{2}$. In particular, they proved that the conjecture is true with no restrictions on $f$ for $\alpha \in [0,\underline{\alpha}_n]$ such that $\underline{\alpha}_n \to 0$ as $n \to \infty$.  For $\alpha \in \left[\frac{1}{2} - \overline{\alpha}_n, \frac{1}{2}\right]$ with $\overline{\alpha}_n \to 0$ as $n \to 0$, they showed that the conjecture holds under the additional assumption that $f$ is unbiased.  They also improved the bound of \cite{Erk96} for unbiased functions $f$, showing that in this case
\[
\MI{f(\bx)}{\by} \leq \frac{\log e}{2}(1-2\alpha)^2 + 9\left(1-\frac{\log e}{2}\right)(1-2\alpha)^4
\]
for $\alpha \in \left[\frac{1}{2}\left(1-\frac{1}{\sqrt{3}}\right),\frac{1}{2}\right]$.  The authors point out that this bound approaches $1-h(\alpha)$ as $\alpha \to \frac{1}{2}$.

\section{A problematic approach to the conjecture}
It is natural to attempt to strengthen the Courtade--Kumar Conjecture by determining the maximum value of $\MI{f(\bx)}{\by}$ among functions of each fixed mean $\mu = \E[f]$.  For example, one might try to prove the equivalent formulation in terms of $\Ent^{\Phi}$ by an induction on~$n$ (or \emph{tensorization}), as discussed in~\cite{Cha04}.  Although the maximizing~$f$ for the original conjecture presumably occurs for~$\mu = 0$, an inductive approach would lead to subfunctions of~$f$ which wouldn't necessarily have mean~$0$.

Indeed, Courtade and Kumar made such a stronger conjecture, discussed in this section.  In discussing this generalization of the problem, we will find it convenient to switch notation, now thinking of $f \co \bn \to \{0,1\}$.

\begin{named}{Courtade--Kumar Lex Conjecture} Fix $n$ and let $(\bx,\by)$ be $\rho$-correlated $n$-bit strings.  Among all functions $f \co \bn \to \{0,1\}$ with a fixed mean $\E[f] = \mu$, the mutual information $\MI{f(\bx)}{\by}$ is maximized when $f$ is ``lex''; i.e., the indicator of the first $\mu 2^n$ points of~$\bn$ in lexicographic ordering.
\end{named}

\begin{remark}  \label{rem:subcube}
    In particular, if $\mu$ is of the form $2^{-k}$ for some integer $0 \leq k \leq n$, the conjecture is that a maximizing~$f$ is an indicator of a $k$-codimensional subcube; equivalently, a logical $k$-AND function.
\end{remark}

If true, this Lex Conjecture would essentially resolve the original conjecture.  We remark that when $f$ is a $k$-AND function as in Remark~\ref{rem:subcube}, it's not hard to calculate that $\MI{f(\bx)}{\by}$ has the simple form $k2^{1-k}(1-h(\alpha))$, making the Lex Conjecture particularly tempting.  Unfortunately, Chandar and Tchamkerten~\cite{CT14} showed that the Lex Conjecture fails.  Specifically, they showed that for each $\alpha$ there exists $k \in \N$ such that $k$-AND functions do \emph{not} maximize $\MI{f(\bx)}{\by}$ among $f \co \bn \to \{0,1\}$ of mean $2^{-k}$ (assuming $n$ is sufficiently large).  In particular, they showed that indicators of (essentially) Hamming spheres do better.

A subsequent version of the Courtade--Kumar paper~\cite{CK14} suggested working around this counterexample by revising the Lex Conjecture to assume that $h(\mu) \geq 1-h(\alpha)$; i.e., that $\mu$ is not too close to $0$ or $1$.  Unfortunately, we show below that this revision does not help.  Indeed, we show that once~$\mu$ is close enough to~$0$ (but still ``constant''), the Lex Conjecture becomes false as $\rho \to 0$ (which is equivalent to $\alpha \to \half$ and hence $1-h(\alpha) \to 0$).

\paragraph{Failure of the Lex Conjecture as $\rho \to 0$.} To see this, first note that among functions $f \co \bn \to \{0,1\}$ of fixed mean $\mu$, maximizing $\MI{f(\bx)}{\by}$ is equivalent to minimizing $\E[h(\T_\rho f(\bx))]$. Recall the Fourier formula
    \[
        \T_\rho f = \mu + \rho f^{=1} + \rho^2 f^{=2} + \rho^3 f^{=3} + \cdots,
    \]
where $f^{=j} = \sum_{|S| = j} \wh{f}(S) \prod_{i \in S} x_i$.  Thinking of $\rho \to 0$, we apply the Taylor expansion to $h(\T_\rho f(x))$ and deduce that it is of the form
\[
    h(\mu) + c_0(\mu) f^{=1}(x) \cdot \rho + \bigl(c_1(\mu) f^{=2}(x) + c_2(\mu) f^{=1}(x)^2 \bigr) \cdot \rho^2 + \bigl(\cdots\bigr) \cdot \rho^3 + \cdots,
\]
where the $c_i(\mu)$'s are certain constants depending only on~$\mu$.  In particular one may check that $c_2(\mu) = -\frac{1}{2 \ln 2 \cdot \mu(1-\mu)} < 0$. Thus when we take the expectation over~$\bx$, we find that minimizing $\E[h(\T_\rho f)]$ (for $\rho$ sufficiently close to~$0$) becomes equivalent to maximizing $\W{1}[f] = \E[f^{=1}(\bx)^2] = \sum_{i=1}^n \wh{f}(\{i\})^2$, the Fourier weight at degree~$1$.  

The question of precisely maximizing the Fourier weight at degree~$1$ among $f \co \bn \to \{0,1\}$ of mean~$\mu$ is a well-known, difficult one.  However, it is a folklore fact that indicators of Hamming balls are superior to logical ANDs (i.e., lex functions) when $\mu$ is sufficiently small.  More precisely, suppose we fix $\mu = 2^{-k}$ for some $k \in \N^+$. Then from~\cite[Props.~5.24,5.25,5.27]{OD14} we have that $\W{1}[\text{AND}_k] = \mu^2 \log(\frac{1}{\mu})$ but that there are Hamming ball indicators $f_n \co \bn \to \{0,1\}$ with
\[
    \E[f_n] \xrightarrow{n \to \infty} \mu, \qquad \W{1}[f_n] \xrightarrow{n \to \infty} \mathcal{U}(\mu)^2 \sim (2 \ln 2) \mu^2 \log(\tfrac{1}{\mu}) \geq 1.386 \mu^2 \log(\tfrac{1}{\mu}).
\]
Here $\mathcal{U}$ denotes the \emph{Gaussian isoperimetric function}.  If $k$ is large enough that $\mathcal{U}(\mu)^2 \geq 1.38\mu^2 \log(\tfrac{1}{\mu})$ then by taking~$n$ large enough and slightly modifying~$f_n$ we can ensure that $\E[f_n] = \mu$ exactly while still retaining $\W{1}[f_n] \geq 1.3 \mu^2 \log(\tfrac{1}{\mu}) = 1.3 \W{1}[\text{AND}_k]$.  Then for~$\rho$ sufficiently close to~$0$ (i.e., $\alpha$ sufficiently close to $\frac12$) we will be able to conclude that $\MI{f_n(\bx)}{\by} > \MI{\text{AND}_k(\bx)}{\by}$.

\section{The problem in continuous settings}
We have shown that resolving the more general conjecture of maximizing $\MI{f(\bx)}{\by}$ among~$f$ of a fixed mean looks to be very difficult in the Boolean setting, since even the problem of maximizing $\W{1}[f]$ among~$f$ of fixed mean is unsolved.  A difficulty with this problem seems to be the lack of effective symmetrization techniques in the discrete setting.  To combat this we propose investigating the Courtade--Kumar problem in natural continuous settings.  

For isoperimetric problems, Gaussian and spherical analogues have been studied extensively.  The appearance of the Gaussian isoperimetric function above suggests that the Courtade-Kumar problem is related to isoperimetric problems and motivates its investigation in Gaussian space and on the sphere.  In addition, one can think of the Gaussian setting as a special case of the original Boolean problem via the Central Limit Theorem.  The study of the Gaussian analogue is further motivated by the frequent use of Gaussian random variables in other areas of information theory, including, for example, in the context of Gaussian channels \cite{CT91}.

In both the spherical and Gaussian settings, the Courtade-Kumar problem can be stated as ``What function maximizes $H(f(\bx)) - H(f(\bx) | \by)$ when $\bx$ and $\by$ are $\rho$-correlated vectors?" for the appropriate notion of $\rho$-correlated vectors.  We consider $0/1$-valued functions $f$, but $\bx$ and $\by$ are drawn from a continuous domain.  We then define $H(f(\bx) | \by) = \E_{y}[H(f(\bx) | \by = y)]$.  For fixed mean $\mu$, we want to find $f$ maximizing $-H(f(\bx) | \by)$.

\paragraph{Gaussian space.}  In this case we define $\bx$ and $\by$ to be \emph{$\rho$-correlated $n$-dimensional standard Gaussian random vectors}.  This means that $\bx$ is a standard $n$-dimensional Gaussian random vector and $\by = \rho \bx + \sqrt{1-\rho^2} \bz$, where $\bz$ is an independent standard $n$-dimensional Gaussian random vector.  Equivalently, the pairs $(\bx_i, \by_i)$ are independent across $1 \leq i \leq n$ and each is distributed as a $2$-dimensional mean-zero Gaussian with covariance matrix $\begin{pmatrix} 1 & \rho \\ \rho & 1 \end{pmatrix}$.  In analogy with $T_{\rho}$, we define the Gaussian noise operator $U_{\rho} f(x) = \E[f(\by) \mid \bx = x]$.  We can then write $-H(f(\bx) | \by) = \E_{\bx \sim \stdgauss^n}[-h(U_{\rho}f(\bx))]$.

We show that halfspaces are most informative in Gaussian space.
\begin{theorem} \label{thm:gauss-most-informative-intro}
Let $f:\R^n \to \{0,1\}$ and let $\bx$ and $\by$ be $\rho$-correlated standard Gaussian random vectors with $0 \leq \rho < 1$.  Then $-H(f(\bx) | \by) \leq -H(\indic{\eta}(\bx) | \by)$, where $\indic{\eta}$ is the indicator of a halfspace $\eta$ such that $\E_{\bx \simÊ\stdgauss^n}[\indic{\eta}(\bx)] = \E_{\bx \simÊ\stdgauss^n}[f(\bx)]$.
\end{theorem}
The theorem shows that for fixed mean $\mu$, halfspaces of mean $\mu$ are optimal.  We expect that halfspaces with mean $1/2$ are optimal overall, but were unable to show this.  We observe that the statement of this theorem is very similar to a general form of Borell's Isoperimetric Theorem \cite{Bor85}:
\begin{theorem}{\textup{\cite{Bor85}}} \label{thm:borell}
Let $f:\R^n \to \{0,1\}$ and $\Psi:\R_{\geq 0} \to \R$ be increasing and convex.  Then
\[
\E_{\bx \simÊ\stdgauss^n}[\Psi(U_{\rho} f(\bx))] \leq \E_{\bx \simÊ\stdgauss^n}[\Psi(U_{\rho} \indic{\eta}(\bx))],
\]
where $\indic{\eta}$ is the indicator function of any halfspace $\eta$ such that $\E_{\bx \simÊ\stdgauss^n}[\indic{\eta}(\bx)] = \E_{\bx \simÊ\stdgauss^n}[f(\bx)]$.
\end{theorem}
Although it may be possible to deduce Theorem~\ref{thm:gauss-most-informative-intro} from Theorem~\ref{thm:borell}, we did not see how to do this.  Our function $-h(x)$ is convex, but is not increasing.

  Our proof follows from the spherical case below via Poincar\'{e}'s limit.  This is the observation that a uniform random point on a high-dimensional sphere projected onto a small number of coordinates looks Gaussian.  The proof idea is from Beckner \cite{Bec92} with details filled in by Carlen and Loss \cite{CL90}.

\paragraph{The sphere.}  In this case we define $\bx$ and $\by$ to be \emph{$\rho$-correlated points on the unit sphere $S^{n-1}$ in $n$ dimensions}.  This means that $\bx$ is a uniformly random point on the surface of $S^{n-1}$ and that $\by$ is the result of a $\ln(1/\rho)$-time Brownian motion on $S^{n-1}$ started at~$\bx$.  Equivalently, $\by$ is defined to be the first point on $S^{n-1}$ hit by a standard $n$-dimensional Brownian motion started from $\rho \bx$.  We denote the corresponding noise operator by $P_{\rho}$.  Then $-H(f(\bx) | \by) = \E_{\bx \sim \stdgauss^n}[-h(P_{\rho}f(\bx))]$.

Once again, halfspaces are most informative.  We write $\bx \sim S^{n-1}$ for $\bx$ drawn uniformly at random from the surface of $S^{n-1}$.
\begin{theorem} \label{thm:spherical-most-informative-intro}
Let $f:S^{n-1} \to \{0,1\}$ and let $\bx$ and $\by$ be $\rho$-correlated points on the unit sphere $S^{n-1}$ with $0 \leq \rho < 1$.  Then $-H(f(\bx) | \by) \leq -H(\indic{\eta}(\bx) | \by)$, where $\indic{\eta}$ is the indicator of a halfspace $\eta$ such that $\E_{\bx \sim S^{n-1}}[\indic{\eta}(\bx)] = \E_{\bx \sim S^{n-1}}[f(\bx)]$.
\end{theorem}
Again, we believe that halfspaces with mean $1/2$ are optimal, but were not able to show this.

The halfspace $\indic{\eta}$ is a symmetrization of the corresponding function $f$.  Rather than directly proving that this symmetrization increases the mutual information, we show that a much simpler notion of symmetrization called polarization increases the mutual information.  The halfspace symmetrization can be thought of as the limit of repeated polarization and we use an argument of Baernstein and Taylor \cite{BT76} to pass from polarizations to halfspaces. 

\paragraph{Noise operators via kernels} Here we mention alternative formulations of $U_{\rho} f$ and $P_{\rho} f$ that we will use below.  In the Gaussian case, the Mehler kernel $U_{\rho}(x,y)$ is defined as
\[
U_{\rho}(x,y) = \frac{1}{(1-\rho^2)^{n/2}}\exp\left(-\frac{\rho^2 \|x\|^2 + 2\rho \la x,y \ra + \rho^2 \|y\|^2}{2(1-\rho^2)}\right).
\]
We can write $U_{\rho}$ in terms of the Mehler kernel: $U_{\rho} f(x) = \E_{\by \sim \stdgauss^n}[U_{\rho}(x,\by) f(\by)]$.

We define the Poisson kernel $P_{\rho}$ in the spherical case:
\[
P_{\rho}(x,y) = \frac{1-\rho^2}{\|x - \rho y\|^n}.
\]
Similarly, $P_{\rho} f(x) = \E_{\by \sim S^{n-1}}[P_{\rho}(x,\by) f(\by)]$.


\section{The spherical case}
Let $S_R^{n-1}$ be the sphere of radius $R$ in $n$ dimensions.  For $x = (x_1,x_2,\ldots,x_n) \in S_R^{n-1}$, the polar angle $\theta_x$ is the angle between $x$ and $r = (R,0,\ldots,0)$.  In other words, $x_1 = R \cos \theta_x$.  Let $\omega_{n-1, R}$ be the uniform probability measure on $S_R^{n-1}$; we will omit the subscripts when they are clear from the context.  Let $C(\theta)$ denote the spherical cap $\{x \in S_R^{n-1} : \theta_x \in [0,\theta) \}$.  For $f:S_R^{n-1} \to \R$, we define the symmetric decreasing rearrangement of $f$ as
\[
\tilde{f}(x) = \inf\{t : \omega(y : f(y) > t) \leq \omega(C(\theta_x))\}.
\]

We will show the following result:
\begin{theorem}
\label{thm:spherical}
Let $m$ be a uniform measure on $S^{n-1}_R$, which may or may not be normalized.  Let $\Psi:\R \to \R$ be a convex, uniformly continuous function and let $f:S_R^{n-1} \to [0,1]$ be integrable.  Let $K:\R \to \R$ be a non-decreasing bounded measurable function.  Then
\[
\int_{S^{n-1}_R} \Psi\left(\int_{S^{n-1}_R} K(\la x, y\ra) f(y) \,dm(y)\right) \,dm(x) \leq \int_{S^{n-1}_R} \Psi\left(\int_{S^{n-1}_R} K(\la x, y\ra) \tilde{f}(y) \,dm(y)\right) \,dm(x).
\]
\end{theorem}
Theorem~\ref{thm:spherical-most-informative-intro} is an immediate corollary.
\begin{proof}[Proof of Theorem~\ref{thm:spherical-most-informative-intro}]
Observe that
\[
-H(f(\bx) | \by) = \E[-h(P_{\rho} f(\bx))] = \int_{S^{n-1}} -h\left(\int_{S^{n-1}} P_{\rho}(x,y) f(y) \,d\omega(y)\right) \,d\omega(x).
\]
Since $P_{\rho}(x,y)$ is a non-decreasing function of $\la x,y \ra$ and $-h$ is convex, Theorem~\ref{thm:spherical} implies that this quantity is upper bounded by
\[
\int_{S^{n-1}} -h(P_{\rho} \wt{f}(x)) d\omega(x) = \E[-h(P_{\rho} \wt{f}(\bx))] = -H(\wt{f}(\bx) | \by).
\]
It is easy to see that $\wt{f} = \indic{\eta}$ for some halfspace $\eta$ such that $\E[f] = \E[\indic{\eta}]$.
\end{proof}

Following Baernstein and Taylor \cite{BT76}, we prove Theorem~\ref{thm:spherical} for a simpler symmetrization called a polarization.  The symmetric decreasing rearrangement can be thought of as the limit of repeated polarizations, so we obtain the desired result.

We now define the polarization operation.  Let $\sigma$ be a hyperplane through the origin that does not pass though $r$.  Let $H^+_{\sigma}$ be the hemisphere defined by $\sigma$ that contains $r$ and let $H^-_{\sigma}$ be the other hemisphere defined by $\sigma$.  For $x \in S_R^n$, we will denote the reflection of $x$ across $\sigma$ as $\sigma x$.  Then the polarization of $f:S_R^{n-1} \to \R$ with respect to $\sigma$ is
\[
f^{\sigma}(x) = \begin{cases}
\max\{f(x),f(\sigma x)\} \quad \text{if $x \in H^+_{\sigma}$} \\
\min\{f(x),f(\sigma x)\} \quad \text{if $x \in H^-_{\sigma}$.}
\end{cases}
\]
To simplify notation, define $Kf(x) = \int_{S^{n-1}_R} K(\la x,y \ra) f(y) \,dm(y)$.  We will prove the following statement:
\begin{theorem} \label{thm:spherical-polarization}
Under the assumptions of Theorem~\ref{thm:spherical},
\[
\int_{S^{n-1}_R} \Psi(Kf(x)) \,dm(x) \leq \int_{S^{n-1}_R} \Psi(Kf^{\sigma}(x)) \,dm(x).
\]
for every hyperplane $\sigma$ passing through the origin that does not contain $r$.
\end{theorem}
As in \cite{BT76}, proving this result for polarizations implies the corresponding result for the symmetric decreasing rearrangement.
\begin{lemma}
\label{lem:polar-to-symm}
Under the assumptions of Theorem~\ref{thm:spherical}, if
\[
\int_{S^{n-1}_R} \Psi(Kf(x)) \,dm(x) \leq \int_{S^{n-1}_R} \Psi(Kf^{\sigma}(x)) \,dm(x),
\]
for every hyperplane $\sigma$ passing through the origin that does not contain $r$, then
\[
\int_{S^{n-1}_R} \Psi(Kf(x)) \,dm(x) \leq \int_{S^{n-1}_R} \Psi(K\wt{f}(x)) \,dm(x).
\]
\end{lemma}
The proof of this lemma exactly follows an argument from \cite{BT76}; we include the proof in Appendix~\ref{sec:polar-to-symm} for completeness.

We will now prove Theorem~\ref{thm:spherical-polarization}.  First, we will need a couple of lemmas about the interaction of these reflections with inner products.
\begin{lemma}
\label{lem:reflection-dot-product-equal}
For $x,y \in S^{n-1}$ and any hyperplane $\sigma$ through the origin, $\la x,y \ra = \la \sigma x, \sigma y \ra$.
\end{lemma}
\begin{proof}
$\sigma x = Ux$ for some unitary matrix $U$.  The lemma follows.
\end{proof}

\begin{lemma}
\label{lem:reflection-dot-product-inequality}
If $x \in H^+_{\sigma}$, then $\la x,y \ra \geq \la \sigma x, y \ra$ for all $y \in H^+_{\sigma}$.  Similarly, if $x \in H^-_{\sigma}$, then $\la x,y \ra \leq \la \sigma x, y \ra$ for all $y \in H^+_{\sigma}$.
\end{lemma}
\begin{proof}
Let $v$ be the unit vector perpendicular to the hyperplane $\sigma$ such that $v \in H^+_{\sigma}$.  Write $x = \alpha_x v + v^{\perp}_x$, where $v^{\perp}_x$ is orthogonal to $v$.  Then $\sigma x = -\alpha_x v + v^{\perp}_x$.  For $x,y \in H^+_{\sigma}$, $\alpha_x, \alpha_y \geq 0$ and we then have that
\[
\la x,y \ra = \alpha_x \alpha_y + \la v^{\perp}_x, v^{\perp}_y \ra \geq -\alpha_x \alpha_y + \la v^{\perp}_x, v^{\perp}_y \ra = \la \sigma x, y \ra.
\]
The proof of the second statement is similar.
\end{proof}

We will also need a lemma about convex functions.
\begin{lemma}
\label{lem:convex-sum}
Let $\Psi:\R \to \R$ be convex and consider $x,y,x',y'$ such that $x+y = x'+y'$ and $\abs{x'-y'} \geq \abs{x-y}$.  Then $\Psi(x) + \Psi(y) \leq \Psi(x') + \Psi(y')$.
\end{lemma}
\begin{proof}
Assume $x' \ne y'$; the result is obvious otherwise.  Without loss of generality, let $y' \geq x'$ and $y \geq x$. It is then clear that $x' \leq x \leq y \leq y'$.

Now let $\lambda = \frac{y'-x}{y'-x'} \in [0,1]$.  Short calculations show that $x = \lambda x' + (1-\lambda) y'$ and $y = (1-\lambda)x' + \lambda y'$.  By convexity,
\begin{align*}
\Psi(x) &= \Psi(\lambda x' + (1-\lambda) y') \leq \lambda \Psi(x') + (1-\lambda)\Psi(y') \\
\Psi(y) &= \Psi((1-\lambda)x'+\lambda y') \leq (1-\lambda)\Psi(x') + \lambda \Psi(y').
\end{align*}
Adding these two inequalities completes the proof of the lemma.
\end{proof}

We now come to the two main lemmas of this section.
\begin{lemma}
\label{lem:sum-equal}
$Kf(x) + Kf(\sigma x) = Kf^{\sigma}(x) + Kf^{\sigma}(\sigma x)$.
\end{lemma}
\begin{proof}
Expanding definitions and using reflections, we can write $Kf(x) + Kf(\sigma x)$ as
\[
\int_{H^+_{\sigma}} K(\la x,y \ra) f(y) + K(\la x,\sigma y \ra) f(\sigma y) + K(\la \sigma x,y \ra) f(y) + K(\la \sigma x,\sigma y \ra) f(\sigma y) \,dm(y).
\]
By Lemma~\ref{lem:reflection-dot-product-equal}, this is equal to $\int_{H^+_{\sigma}} (K(\la x,y \ra) + K(\la \sigma x,y \ra))(f(y) + f(\sigma y)) \,dm(y)$.

Similarly,
\[
Kf^{\sigma}(x) + Kf^{\sigma}(\sigma x) = \int_{H^+_{\sigma}} (K(\la x,y \ra) + K(\la \sigma x,y \ra))(f^{\sigma}(y) + f^{\sigma}(\sigma y)) \,dm(y).
\]
By the definition of $f^{\sigma}$, $f(y) + f(\sigma y) = f^{\sigma}(y) + f^{\sigma}(\sigma y)$, so the two integrands are equal and the lemma follows.
\end{proof}

\begin{lemma}
\label{lem:diff-bigger}
$\abs{Kf^{\sigma}(x) - Kf^{\sigma}(\sigma x)} \geq \abs{Kf(x) - Kf(\sigma x)}$.
\end{lemma}
\begin{proof}
By similar calculations to those in the proof of the previous lemma,
\begin{align*}
Kf^{\sigma}(x) - Kf^{\sigma}(\sigma x) &= \int_{H^+_{\sigma}} (K(\la x,y \ra) - K(\la \sigma x,y \ra))(f^{\sigma}(y) - f^{\sigma}(\sigma y)) \,dm(y) \\
Kf(x) - Kf(\sigma x) &= \int_{H^+_{\sigma}} (K(\la x,y \ra) - K(\la \sigma x,y \ra))(f(y) - f(\sigma y)) \,dm(y).
\end{align*}

First, observe that $f^{\sigma}(y) - f^{\sigma}(\sigma y) = \abs{f(y)-f(\sigma y)}$ for $y \in H^+_{\sigma}$.  Next, note that for fixed $x$, $K(\la x,y \ra)-K(\la \sigma x,y \ra)$ has the same sign for all $y \in H^+_{\sigma}$.  Indeed, if $x \in H^+_{\sigma}$, then $K(\la x,y \ra) \geq K(\la \sigma x,y \ra)$ for all $y \in H^+_{\sigma}$ by Lemma~\ref{lem:reflection-dot-product-inequality}.  Likewise, if $x \in H^-_{\sigma}$, then $K(\la x,y \ra) \leq K(\la \sigma x,y \ra)$ for all $y \in H^+_{\sigma}$. We can therefore write
\begin{align*}
\abs{Kf^{\sigma}(x) - Kf^{\sigma}(\sigma x)} &= \int_{H^+_{\sigma}} \abs{(K(\la x,y \ra) - K(\la \sigma x,y \ra))(f(y) - f(\sigma y))} \,dm(y) \\
&\geq \abs{\int_{H^+_{\sigma}} (K(\la x,y \ra) - K(\la \sigma x,y \ra))(f(y) - f(\sigma y)) \,dm(y)}. \qedhere
\end{align*}
\end{proof}

Using Lemmas~\ref{lem:sum-equal} and \ref{lem:diff-bigger}, we can now complete the proof of the theorem.
\begin{proof}[Proof of Theorem \ref{thm:spherical}]
\begin{align*}
\int_{S^{n-1}_R} \Psi(Kf(x)) \,dm(x) &= \int_{H^+_{\sigma}} \Psi(Kf(x)) + \Psi(Kf(\sigma x)) \,dm(x) \\
&\leq \int_{H^+_{\sigma}} \Psi(Kf^{\sigma}(x)) + \Psi(Kf^{\sigma}(\sigma x)) \,dm(x) \quad \text{by Lemmas~\ref{lem:sum-equal},~\ref{lem:diff-bigger},~and~\ref{lem:convex-sum}}\\
&= \int_{S^{n-1}_R} \Psi(Kf^{\sigma}(x)) \,dm(x). \qedhere
\end{align*}
\end{proof}

\section{The Gaussian case} \label{sec:gaussian}
In this section, we will use Theorem~\ref{thm:spherical} to prove that halfspaces are most informative in Gaussian space.  Let $\gamma$ be the standard Gaussian measure on $\R^n$, which has density $\frac{1}{(2\pi)^{n/2}} \exp\left(-\frac{1}{2}\|x\|^2\right)$.

\begin{theorem}
\label{thm:gaussian}
Let $\Psi:\R \to \R$ be convex, bounded, and uniformly continuous and let $f:\R^n \to \{0,1\}$.  Let $\rho \in [0,1)$.  Then
\[
\int_{\mathbb{R}^n} \Psi(U_{\rho} f(x)) \,d\gamma(x) \leq \int_{\mathbb{R}^n} \Psi(U_{\rho} \indic{\eta}(x)) \,d\gamma(x),
\]
where $\indic{\eta}$ is the indicator function of some halfspace $\eta$ such that $\E_{\bx \in \stdgauss^n}[f(\bx)] = \E_{\bx \in \stdgauss^n}[\indic{\eta}(\bx)]$.
\end{theorem}
Taking $\Psi = -h$, this immediately implies Theorem~\ref{thm:gauss-most-informative-intro}.  To reduce clutter, we will write drop the factor of $\frac{1}{(2\pi)^{n/2}}$ and write $d\gamma(x) = \exp\left(-\frac{1}{2}\|x\|^2\right)\,dx$ for the rest of this section.

\subsection{The proof idea}
First, we give the intuition behind the proof.  For $\bu$ drawn uniformly at random from $S^{N-1}_{\sqrt{N}}$, the projection of $\bu$ onto its first $n$ coordinates is close to being distributed as an $n$-dimensional Gaussian for large $N$.  This well-known fact is sometimes called Poincar\'{e}'s observation.  We can use this idea to transfer results for the sphere to Gaussian space as was done in \cite{Bec92, CL90}.

To make this plan more concrete, observe that we can write $u \in S_R^{N-1}$ as 
\begin{equation}
\label{eqn:sphere-decomp}
u = \left(x, \left(1-\frac{\norm{x}^2}{R}\right)^{1/2} v\right),
\end{equation}
where $x \in B^n_R$ and $v \in S^{N-n-1}_R$.  Given $f:\R^n \to \R$, we then define $f^{\mathrm{ext}}$ to be the extension of $f$ to $S^{N-1}_R$.  More formally, we define $f^{\mathrm{ext}} \co S^{N-1}_R \to \R$ such that $f^{\mathrm{ext}}(u) = f(u_1,u_2,\ldots,u_n)$.  The idea of the proof is to show the desired inequality involving $f$ on the sphere for $f^{\mathrm{ext}}$ and then take the limit as $N$ increases to derive the corresponding inequality for $f$.

We now give a simple example: For bounded $f \co \R^n \to \R$, the expectation of $f^{\mathrm{ext}}$ on $S^{N-1}_R$ converges to the expectation of $f$ in Gaussian space.  First, we give a formula for integrating over the sphere according to the decomposition in \eqref{eqn:sphere-decomp}.  Let $s_{N-1,R}$ be the uniform surface measure on $S^{N-1}_R$.  We will suppress the subscripts, as they will be clear from the context.
\begin{lemma}
\label{lem:decomp-sphere-int}
Let $g:S_R^{N-1} \to \R$. Then
\[
\int_{S_R^{N-1}} g(u) \,ds(u) = \int_{B_R^n} \int_{S_R^{N-n-1}} g(x,v) \left(1-\frac{\norm{x}^2}{R^2}\right)^{\frac{N-n-3}{2}} \,ds(v) \,dx.
\]
\end{lemma}
This is essentially shown in, e.g., \cite{ABR01}.

For the rest of this paper, set $R = \sqrt{N-n-3}$.  Then observe that
\[
\lim_{N \to \infty} \left(1-\frac{\norm{x}^2}{R^2}\right)^{\frac{N-n-3}{2}} \,dx = \exp\left(- \frac{\norm{x}^2}{2}\right) \,dx = d\gamma(x).
\]
Together with Lemma~\ref{lem:decomp-sphere-int}, this implies that
\begin{equation*} \label{eqn:same-mean}
\lim_{N \to \infty} \int_{S^{N-1}_R} f^{\mathrm{ext}}(u) \,d\omega(u) = \int_{\R^n} f(x) \,d\gamma(x).
\end{equation*}
The proof of Theorem~\ref{thm:gaussian} is not quite so simple: the use of the noise operator raises technical complications.  However, Carlen and Loss \cite{CL90} showed how to overcome these difficulties and pass from inequalities involving the spherical noise operator to inequalities involving the Gaussian noise operator.  We largely follow their treatment, introducing a ``Poisson-like" kernel $Q_{\rho}$ such that $\lim_{N \to \infty} \int \Psi(Q_{\rho}f^{\mathrm{ext}}(u))\,d\omega(u) = \int \Psi(U_{\rho}f(x))\,d\gamma(x)$ and then using Theorem~\ref{thm:spherical} to show that $\int \Psi(Q_{\rho}f^{\mathrm{ext}}(u))\,d\omega(u) \leq \int \Psi(Q_{\rho}\indic{\eta}^{\mathrm{ext}}(u))\,d\omega(u)$.

\subsection{Rewriting a ``Poisson-like" kernel in terms of a ``Mehler-like" kernel}
Following \cite{CL90}, we will construct $Q_{\rho}$ on $S^{N-1}_R \times S^{N-1}_R$ that converges to the Mehler kernel as $N$ increases.  Thinking of $S_R^{N-1}$ as the product of $B_R^{n}$ and $S^{N-n-1}_R$ as in \eqref{eqn:sphere-decomp}, $Q_{\rho}$ will factor into $U_{N,\rho} \cdot P_{\rho'}$ such that $U_{N,\rho} \co B_R^{n} \times B_R^{n} \to \R$ converges to the Mehler kernel and $P_{\rho'} \co S^{N-n-1}_R \times S^{N-n-1}_R \to \R$ is a Poisson kernel that integrates to $1$.

We will now give formal statements of these ideas.  The lemmas in this section are essentially given in \cite{CL90}; we include proofs in Appendix~\ref{sec:gaussian-omitted}.  Recall that $\rho \in [0,1)$.  First, define $Q_{\rho}:S^{N-1}_R \times S^{N-1}_R \to \R$ so that
\[
Q_{\rho}(u, v) = \frac{R(1-\rho^2)^{1-n/2}}{|S^{N-n-1}|\norm{u-\rho v}^{N-n}},
\]
where $|S^{N-n-1}|$ is the surface area of $S^{N-n-1}$.  The ``Mehler kernel" factor of this quantity is
\[
U_{\rho,N}(y,z) = \frac{(1-\rho^2)^{1-n/2}}{(1-r^2(y,z))A(y,z)^{\frac{N-n}{2}}}.
\]
where
\begin{align*}
A(y,z) &= \frac{1+\rho^2-\frac{2\rho}{R^2} \la y,z\rangle}{2} + \sqrt{\left(\frac{1+\rho^2-\frac{2\rho}{R^2} \la y,z \ra}{2}\right)^2-\rho^2\left(1-\frac{\norm{y}^2}{R^2}\right)\left(1-\frac{\norm{z}^2}{R^2}\right)} \quad \text{and}\\
r(y,z) &= \frac{\rho\left(1-\frac{\norm{y}^2}{R^2}\right)^{1/2}\left(1-\frac{\norm{z}^2}{R^2}\right)^{1/2}}{A(y,z)}.
\end{align*}
The next lemma shows that $Q_{\rho}$ can be written as a product of $U_{\rho,N}(y,z)$ and a Poisson kernel.
\begin{lemma} \label{lem:factor-kernel}
Let $u = \left(y, \left(1-\frac{\norm{y}^2}{R^2}\right)^{1/2} w\right) \in S^{N-1}_R$ such that $y \in B_R^{n}$ and $w \in S^{N-n-1}_R$ as in \eqref{eqn:sphere-decomp}.  Likewise, let $v = \left(z, \left(1-\frac{\norm{z}^2}{R^2}\right)^{1/2} x\right) \in S^{N-1}_R$ such that $z \in B_R^{n}$ and $x \in S^{N-n-1}_R$.  Then
\[
Q_{\rho}(u, v) = U_{\rho,N}(y,z)  \frac{R(1-r^2)}{|S^{N-n-1}|\norm{w - r x}^{N-n}}
\]
and $r \in [0,1)$.
\end{lemma}
We address the Mehler and Poisson factors in turn. As $N$ goes to $\infty$, $U_{\rho,N}(y,z)$ converges to the Mehler kernel.
\begin{lemma} \label{lem:mehler-limit}
$\lim_{N \to \infty} U_{\rho,N}(y,z) = U_{\rho}(y,z)$.
\end{lemma}

The Poisson kernel factor integrates to $1$.
\begin{lemma}
\label{lem:poisson-int}
$\int_{S^{N-n-1}_R} \frac{R(1-r^2)}{|S^{N-n-1}|\norm{w - r x}^{N-n}} \,ds(x) = 1$.
\end{lemma}

Define $Q_{\rho} f(u) = \int_{S^{N-1}_R} Q_{\rho}(u,v) f(v)\,d\omega(v)$ and
\[
U_{\rho,N}f(y) = \int_{\R^n} \indic{\norm{y} \leq R} U_{\rho,N}(y,z) f(z) \left(1-\frac{\norm{z}^2}{R^2}\right)^{\frac{N-n-3}{2}} \,dz.
\]
In the main lemma of this section, we will use the above lemmas to rewrite the spherical quantity $\int \Psi(Q_{\rho}f(u)) \,d\omega(u)$ in terms of $U_{\rho,N}$.

\begin{lemma} \label{lem:simplify-spherical}
\[
\int_{S^{N-1}_R} \Psi\left(\abs{S^{N-1}_R} \cdot Q_{\rho} f^{\mathrm{ext}}(u)\right) \,d\omega(u) = \frac{\abs{S^{N-n-1}_R}}{\abs{S^{N-1}_R}} \int_{\mathbb{R}^n} \indic{\norm{y} \leq R} \, \Psi\left(U_{\rho,N}f(y)\right) \left(1-\frac{\norm{y}^2}{R^2}\right)^{\frac{N-n-3}{2}} \,dy.
\]
\end{lemma}
\begin{proof}
Lemmas~\ref{lem:decomp-sphere-int}~and~\ref{lem:factor-kernel} imply that $Q_{\rho} f^{\mathrm{ext}}(u)$ is equal to
\[
\frac{1}{\abs{S^{N-1}_R}}\int_{B_R^n} f(z) U_{\rho,N}(y,z) \left(\int_{S_R^{N-n-1}} \frac{R(1-r^2)}{|S^{N-n-1}|\norm{w-rx}^{N-n}} \,ds(x)\right) \left(1-\frac{\norm{z}^2}{R^2}\right)^{\frac{N-n-3}{2}} \,dz.
\]
Lemma~\ref{lem:poisson-int} then shows that $Q_{\rho} f^{\mathrm{ext}}(u) = \frac{1}{\abs{S^{N-1}_R}}U_{\rho,N}f(x)$.  Applying Lemma~\ref{lem:decomp-sphere-int} to the outer integral completes the proof.
\end{proof}

\subsection{Passing from the sphere to Gaussian space}
Using the previous section, we now prove our main lemma.  It essentially states that the spherical quantity $\int \Psi(Q_{\rho}f(u)) \,d\omega(u)$ converges to the Gaussian quantity $\int \Psi(U_{\rho} f(y)) \,d\gamma(y)$ that we would like to bound.
\begin{lemma} \label{lem:sphere-limit}
$\lim_{N \to \infty} \int_{S^{N-1}_R} \Psi\left(\abs{S^{N-1}_R} \cdot Q_{\rho} f^{\mathrm{ext}}(u)\right) \,d\omega(u) = \int_{\mathbb{R}^n} \Psi(U_{\rho} f(y)) \,d\gamma(y)$.
\end{lemma}
To prove this lemma, we will need an additional technical lemma given in \cite{CL90}.
\begin{lemma} \label{lem:A-bound}
$\left(1-\frac{\norm{y}^2}{R^2}\right)^{1/2} \left(1-\frac{\norm{z}^2}{R^2}\right)^{1/2} \leq A(y,x)$.
\end{lemma}
We give a proof of this lemma in Appendix~\ref{sec:gaussian-omitted}.

\begin{proof}[Proof of Lemma~\ref{lem:sphere-limit}]
By Lemma~\ref{lem:simplify-spherical}, it suffices to show that
\[
\lim_{N \to \infty} \int_{\mathbb{R}^n} \indic{\norm{y} \leq R} \, \Psi\left(U_{\rho,N}f(y)\right) \left(1-\frac{\norm{y}^2}{R^2}\right)^{\frac{N-n-3}{2}} \,dy = \int_{\mathbb{R}^n} \Psi(U_{\rho} f(y)) \,d\gamma(y).
\]
First, we prove that $\lim_{N \to \infty} U_{\rho,N}f(y) = U_{\rho} f(y)$.  For each $y,z \in \R^n$, Lemma~\ref{lem:mehler-limit} implies that
\[
\lim_{N \to \infty} \indic{\norm{y} \leq R} U_{\rho,N}(y,z) f(z) \left(1-\frac{\norm{z}^2}{R^2}\right)^{\frac{N-n-3}{2}} = U_{\rho}(y,z) f(z) \exp\left(-\frac{1}{2}\|z\|^2\right).
\]
We then wish to upper bound $\abs{\indic{\norm{y} \leq R} U_{\rho,N}(y,z) f(z) \left(1-\frac{\norm{z}^2}{R^2}\right)^{\frac{N-n-3}{2}}}$ by an integrable function so we can apply dominated convergence.  Lemma~\ref{lem:A-bound} implies that $r \leq \rho$ and, using the definition of $U_{\rho,N}$, we see that
\[
\abs{\indic{\norm{y} \leq R} U_{\rho,N}(y,z) f(z) \left(1-\frac{\norm{z}^2}{R^2}\right)^{\frac{N-n-3}{2}}} \leq \frac{\left(1-\frac{\norm{z}^2}{R^2}\right)^{\frac{N-n-3}{2}}}{(1-\rho^2)^{n/2}A^{\frac{N-n}{2}}}.
\]
Applying Lemma~\ref{lem:A-bound} again shows that the right hand side is at most $c \exp\left(\frac{\norm{y}^2}{4}\right) \exp\left(-\frac{\norm{z}^2}{4}\right)$ for some $c$ that does not depend on $z$ or $N$.  For a given $y$, this is integrable; dominated convergence then implies that $\lim_{N \to \infty} U_{\rho,N}f(y) = U_{\rho} f(y)$.  Since $\Psi$ is uniformly continuous, we exchange the limit and the application of $\Psi$.  Since $\Psi$ is bounded, we can apply dominated convergence to the outer integral to complete the proof.
\end{proof}

We can now prove Theorem~\ref{thm:gaussian}.
\begin{proof}[Proof of Theorem~\ref{thm:gaussian}]
By Theorem~\ref{thm:spherical},
\[
\int_{S^{N-1}_R} \Psi\left(\abs{S^{N-1}_R} \cdot Q_{\rho} f^{\mathrm{ext}}(u)\right) \,d\omega(u) \leq \int_{S^{N-1}_R} \Psi\left(\abs{S^{N-1}_R} \cdot Q_{\rho} \wt{f^{\mathrm{ext}}}(u)\right) \,d\omega(u).
\]
Since $f^{\mathrm{ext}}$ is $0/1$-valued, $\wt{f^{\mathrm{ext}}}$ is the indicator function $\indic{\eta}$ of a halfspace $\eta$.  By symmetry, we assume that $\eta = \{u \in \R^N : u_1 \geq t\}$ for some $t \in \R$.  Then $h$ depends only on the first coordinate of $u$ and $\indic{\eta} = \indic{\eta'}^{\mathrm{ext}}$, where $\eta'$ is the halfspace $\{u \in \R^n : u_1 \geq t\}$.  Using Lemma~\ref{lem:sphere-limit} to take the limit on both sides, we obtain $\int_{\mathbb{R}^n} \Psi(U_{\rho} f(y)) \,d\gamma(y) \leq \int_{\mathbb{R}^n} \Psi(U_{\rho} \indic{\eta'}(y)) \,d\gamma(y)$.

It remains to show that $\E_{\bx \in \stdgauss^n}[f(\bx)] = \E_{\bx \in \stdgauss^n}[\indic{\eta'}(\bx)]$.  To see this, observe that $\int_{S^{N-1}_R} f^{\mathrm{ext}}(u)\,d\omega(u) = \int_{S^{N-1}_R} \indic{\eta'}^{\mathrm{ext}}(u)\,d\omega(u)$.  The result then follows from \eqref{eqn:same-mean}.
\end{proof}

\subsection*{Acknowledgments}
The second-named author would like to thank Eric Blais, Ankit Garg, and Oded Regev for helpful discussions, as well as the Bo\u{g}azi\c{c}i University Computer Engineering Department for their hospitality.

\bibliographystyle{alpha}
\bibliography{odonnell-bib,witmer}

\appendix
\section{From polarizations to the symmetric decreasing rearrangement}
\label{sec:polar-to-symm}

In this section, we give a proof of Lemma~\ref{lem:polar-to-symm}, which was essentially proven by Baernstein and Taylor \cite{BT76}.  Our setting is very slightly different, but no new techniques are required and the proof exactly follows the outline of \cite{BT76}.
\begin{customlem}{\ref{lem:polar-to-symm}}
Let $m$ be a uniform measure on $S^n_R$, which may or may not be normalized.  Let $\Psi:\R \to \R$ be a convex, uniformly continuous function and let $f:S_R^n \to [0,1]$ be integrable.  Let $K:\R \to \R$ be a non-decreasing bounded measurable function.  If
\[
\int_{S^n_R} \Psi(Kf(x)) \,dm(x) \leq \int_{S^n_R} \Psi(Kf^{\sigma}(x)) \,dm(x),
\]
for every hyperplane $\sigma$ passing through the origin that does not contain $r = (R,0,\ldots,0)$, then
\[
\int_{S^n_R} \Psi(Kf(x)) \,dm(x) \leq \int_{S^n_R} \Psi(K\wt{f}(x)) \,dm(x).
\]
\end{customlem}

\begin{proof}
For brevity, define $J(f) = \int_{S^n_R} \Psi(Kf(x)) \,dm(x)$.  As described in \cite{BT76}, it suffices to consider continuous functions $f$: For any $f \in L^1(S^n_R)$ there a sequence of continuous functions $\{f_i\}$ converging to $f$ in the $L_1$ norm. Let $\calC(S^n_R)$ be the set of continuous functions on $S^n_R$; $\calC(S^n_R)$ is complete under the supremum norm.  Recall the definition of the modulus of continuity:
\[
\omega(\delta,f) = \sup\{|f(x)-f(y)| : |x-y| \leq \delta, x,y \in S^n_R \}.
\]
We can then define
\[
\calP = \{F \in \calC(S^n_R) :  \omega(\cdot,F) \leq \omega(\cdot,f) \text{, }\wt{F} = \wt{f}  \text{, and } J(f) \leq J(F)\}.
\]
Observe that $\calP$ is nonempty: it contains $f^{\sigma}$ for all hyperplanes $\sigma$ through the origin.  The fact that the modulus of continuity decreases under polarizations and $\wt{f} = \wt{f^{\sigma}}$ is given in \cite[Lemma~1]{BT76}.  To prove the lemma, it suffices to show that $\wt{f} \in \calP$.  Assume for a contradiction that $\wt{f} \notin \calP$.   Consider
\[
D(F) = \int_{S^n_R} (F - \wt{f})^2\,dm.
\]
We will derive a contradiction by showing that for any function $h$ that minimizes $D$ on $\calP$ with $h \ne \wt{f}$, we can find another function $h'$ such that $D(h') < D(h)$.  To do this, we first need to show that $D$ attains a minimum value on $\calP$ using the Extreme Value Theorem.  In order to use this theorem, we need to show that $\calP$ is compact and $D$ is continuous.

\begin{claim}
$\calP$ is compact under the supremum norm.
\end{claim}
\begin{proof}
We first use the Arzel\`{a}-Ascoli Theorem to show that $\calP$ is relatively compact and then show that the limit of any convergent sequence of functions in $\calP$ is also $\calP$.

To apply the Arzel\`{a}-Ascoli Theorem, we need $\calP$ to be equicontinuous and uniformly bounded.  Equicontinuity is immediate from the definition of $\calP$.  To see that $\calP$ is uniformly bounded, observe that for any $F \in \calP$, it holds this $\abs{F} \leq \sup_{x \in S^n_R}\{|\wt{f}(x)|\}$.  This follows from continuity of $F$ and $\wt{F} = \wt{f}$.  Since $\wt{f} \in L^1(S^n_R)$, it is bounded and thus $\calP$ is uniformly bounded.

It remains to show that the limit of any convergent sequence of functions in $\calP$ is also in $\calP$.  Let $\{g_i\}_{i \in \N}$ be a convergent sequence in $\calP$ and let $\lim_{i \to \infty} g_i = g$.  Since $\calC(S^n_R)$ is complete, it suffices to show that $\omega(\cdot,g) \leq \omega(\cdot,f)$, $\wt{g} = \wt{f}$, and $J(f) \leq J(g)$.  It is clear that $\omega(\cdot,g) \leq \omega(\cdot,f)$ holds.

To see that $\wt{g} = \wt{f}$, assume for a contradiction that $\wt{g}(x) > \wt{f}(x)$; this is without loss of generality.  Then there exist $t  \in \R$ and $\eps > 0$ such that $m(x : g(x) > t + \eps) > m(x : f(x) > t)$.  The right hand side is equal to $m(x : g_i(x) > t)$ for all $i$ since $\wt{g_i} = \wt{f}$.  Then for all $i$, there exists $x$ such that $g(x) - g_i(x) > \eps$.  The contradicts convergence of the $g_i$'s in the supremum norm.

Lastly, we show that $J(f) \leq J(g)$.  Note that the $g_i$'s are uniformly bounded.  We can then apply dominated convergence and use uniform continuity of $\Psi$ to deduce that $\lim_{i \to \infty} J(g_i) = J(g)$.  Since $J(f) \leq J(g_i)$, it must be the case that $J(f) \leq J(g)$.
\end{proof}

\begin{claim}
$D$ is continuous.
\end{claim}
\begin{proof}
Observe that
\[
\abs{D(F) - D(G)} = \abs{\int_{S^n_R} (F-G)(F+G+2\wt{f})\,dm} \leq \sup_{x \in S^n_R}\abs{F(x)-G(x)} \int_{S^n_R} |F+G+2\wt{f}| \,dm.
\]
Since $F$, $G$, and $\wt{f}$ are bounded, $\int_{S^n_R} |F+G+2\wt{f}| \,dm$ is bounded and $\abs{D(F) - D(G)}$ goes to $0$ as the supremum norm $\sup_{x \in S^n_R}\abs{F(x)-G(x)}$ goes to $0$.
\end{proof}
Using these two claims, the Extreme Value Theorem implies that $D$ attains a minimum value on $\calP$.  Let $h \ne \wt{f}$ be a minimizing function in $\calP$.  Now we will derive a contradiction by exhibiting a function $h'$ in $\calP$ such that $D(h') < D(h)$. We will set $h' = h^{\sigma}$ for an appropriately chosen hyperplane $\sigma$.

\begin{claim}
There exists a hyperplane $\sigma$ through the origin and a set $B \subseteq H^+_{\sigma}$ of positive measure such that
\[
\wt{f}(x) > \wt{f}(\sigma x) \text{ and } h(\sigma x) > h(x)
\]
for all $x \in B$.
\end{claim}
\begin{proof}
Since $\wt{h} = \wt{f}$ but $h \ne \wt{f}$, $h$ must not be symmetric decreasing.  That is, there must exist some $t$ such that $E = \{x : h(x) > t\}$ is not equal to $C = \{x : \wt{f}(x) > t\}$.  We know that $\wt{f}$ and $h$ are continuous and that $m(E) = m(C)$, so both $E \setminus C$ and $C \setminus E$ have positive measure.  Let $x$ be density point of $E \setminus C$ and $y$ be a density point of $C \setminus E$.  Let $\sigma$ be the hyperplane through the origin such that $\sigma x = y$.  Then $\wt{f}(y) > t \geq \wt{f}(x)$, so $r \notin \sigma$ and $y \in H^+_{\sigma}$.  Define $B = H^+_{\sigma} \cap (C \setminus E) \cap \sigma(E \setminus C)$.  By considering a small neighborhood around $y$ and its reflection under $\sigma$, we see that $B$ has positive measure.  Then for $x \in B$ it holds that $\wt{f}(x) > \wt{f}(\sigma x)$ and $h(\sigma x) > h(x)$.
\end{proof}

\begin{claim}
\[
\int_{S^n_R} h \wt{f} \,dm < \int_{S^n_R} h^{\sigma} \wt{f} \,dm
\]
\end{claim}
\begin{proof}
Lemma~\ref{lem:reflection-dot-product-inequality} shows that $\la x,r \ra \geq \la \sigma x,r \ra$ for all $x \in H^+_{\sigma}$.  Since $\la x,r \ra = R^2 \cos \theta_x$, $\wt{f}$ is an increasing function of $\la x,r \ra$ and so $\wt{f}(x) \geq \wt{f}(\sigma x)$ for $x \in H^+_{\sigma}$.  By definition, $h^{\sigma}(x) \geq h^{\sigma}(\sigma x)$ for $x \in H^+_{\sigma}$.  For $a_1,a_2,b_1,b_2 \in \R$ with $a_1 \geq a_2$ and $b_1 \geq b_2$, it is easy to show that $a_1 b_2 + a_2 b_1 \leq a_1 b_1 + a_2 b_2$, with strict inequality if $a_1 > a_2$ and $b_1 > b_2$.  In our case, this implies that $h(x)\wt{f}(x) + h(\sigma x) \wt{f}(\sigma x) \leq h^{\sigma}(x) \wt{f}(x) + h^{\sigma}(\sigma x) \wt{f}(\sigma x)$ for all $x \in H^+_{\sigma} \setminus B$ and $h(x)\wt{f}(x) + h(\sigma x) \wt{f}(\sigma x) < h^{\sigma}(x) \wt{f}(x) + h^{\sigma}(\sigma x) \wt{f}(\sigma x)$ for all $x \in B$.  The claim follows:
\begin{align*}
\int_{S^n_R} h(x) \wt{f}(x) \,dm(x) &= \int_{H^+_{\sigma}} h(x)\wt{f}(x) + h(\sigma x) \wt{f}(\sigma x) \,dm(x) \\
&< \int_{H^+_{\sigma}} h^{\sigma}(x) \wt{f}(x) + h^{\sigma}(\sigma x) \wt{f}(\sigma x) \,dm(x) \\
&= \int_{S^n_R} h^{\sigma}(x) \wt{f}(x) \,dm. \qedhere
\end{align*}
\end{proof}

Using this claim, we can complete the proof.  Note that $h$ and $h^{\sigma}$ have the same $L^2$ norm. Then
\[
D(h) = \int (h- \wt{f})^2 \,dm = \int h^2 - 2 h \wt{f} + \wt{f}^2 \,dm > \int (h^{\sigma})^2 - 2 h^{\sigma} \wt{f} + \wt{f}^2 \,dm = \int (h^{\sigma} - \wt{f})^2 \,dm = D(h'),
\]
which is a contradiction.
\end{proof}

\section{Proofs omitted from Section~\ref{sec:gaussian}} \label{sec:gaussian-omitted}
The proofs in this section follow those of Carlen and Loss \cite{CL90}.  Recall the following definitions:
\begin{align*}
Q_{\rho}(u, v) &= \frac{R(1-\rho^2)^{1-n/2}}{|S^{N-n-1}|\norm{u-\rho v}^{N-n}} \\
U_{\rho,N}(y,z) &= \frac{(1-\rho^2)^{1-n/2}}{(1-r^2(y,z))A(y,z)^{\frac{N-n}{2}}}.
\end{align*}
where
\begin{align*}
A(y,z) &= \frac{1+\rho^2-\frac{2\rho}{R^2} \la y,z\rangle}{2} + \sqrt{\left(\frac{1+\rho^2-\frac{2\rho}{R^2} \la y,z \ra}{2}\right)^2-\rho^2\left(1-\frac{\norm{y}^2}{R^2}\right)\left(1-\frac{\norm{z}^2}{R^2}\right)} \quad \text{and}\\
r(y,z) &= \frac{\rho\left(1-\frac{\norm{y}^2}{R^2}\right)^{1/2}\left(1-\frac{\norm{z}^2}{R^2}\right)^{1/2}}{A(y,z)}.
\end{align*}
As above, we set $R = \sqrt{N-n-3}$ and define $|S^{N-1}|$ to be the surface area of $S^{N-1}$.

\subsection{Proof of Lemma~\ref{lem:factor-kernel}}
\begin{customlem}{\ref{lem:factor-kernel}}
Let $u = \left(y, \left(1-\frac{\norm{y}^2}{R^2}\right)^{1/2} w\right) \in S^{N-1}_R$ such that $y \in B_R^{n}$ and $w \in S^{N-n-1}_R$ as in \eqref{eqn:sphere-decomp}.  Likewise, let $v = \left(z, \left(1-\frac{\norm{z}^2}{R^2}\right)^{1/2} x\right) \in S^{N-1}_R$ such that $z \in B_R^{n}$ and $x \in S^{N-n-1}_R$.  Then
\[
Q_{\rho}(u, v) = U_{\rho,N}(y,z)  \frac{R(1-r^2)}{|S^{N-n-1}|\norm{w - r x}^{N-n}}
\]
and $r \in [0,1)$.
\end{customlem}
The proof is outlined in \cite{CL90}. 
\begin{proof}
We want to find $A(x,y)$ and $r(x,y)$ such that
\[
\norm{u - \rho v}^2 = A\norm{w - r x}^2.
\]
Since $\norm{w} = \norm{x} = R$, the left hand side is
\begin{align*}
\norm{u - \rho v}^2 &= \norm{y-\rho z}^2 + \norm{\left(1-\frac{\norm{y}^2}{R^2}\right)^{1/2} w - \left(1-\frac{\norm{z}^2}{R^2}\right)^{1/2} \rho x}^2 \\
&= R^2\left(1+\rho^2-\frac{2\rho}{R^2}\la y,z \ra\right) - 2\rho\left(1-\frac{\norm{y}^2}{R^2}\right)^{1/2} \left(1-\frac{\norm{z}^2}{R^2}\right)^{1/2} \la w, x \ra.
\end{align*}
The right hand side is
\[
A\norm{w - r x}^2 = AR^2(1+r^2) - 2Ar\la w, x \ra.
\]
Setting
\[
2Ar = 2\rho\left(1-\frac{\norm{y}^2}{R^2}\right)^{1/2} \left(1-\frac{\norm{z}^2}{R^2}\right)^{1/2},
\]
we get that
\[
r = \frac{\rho\left(1-\frac{\norm{y}^2}{R^2}\right)^{1/2} \left(1-\frac{\norm{z}^2}{R^2}\right)^{1/2}}{A}.
\]
Setting
\[
AR^2(1+r^2) = R^2\left(1+\rho^2-\frac{2\rho}{R^2}\la y,z \ra\right)
\]
and substituting in the above value for $s$, we get the equation
\[
A^2 - \left(1+\rho^2-\frac{2\rho}{R^2} \la y,z \ra \right) A + \rho^2 \left(1-\frac{\norm{y}^2}{R^2}\right) \left(1-\frac{\norm{z}^2}{R^2}\right) = 0.
\]
Solving, we obtain
\[
A = \frac{1+\rho^2-\frac{2\rho}{R^2} \la y,z\rangle}{2} + \sqrt{\left(\frac{1+\rho^2-\frac{2\rho}{R^2} \la y,z \ra}{2}\right)^2-\rho^2\left(1-\frac{\norm{y}^2}{R^2}\right)\left(1-\frac{\norm{z}^2}{R^2}\right)}.
\]
So we have that
\begin{align*}
Q_{\rho}(u, v) &= \frac{R(1-\rho^2)^{1-n/2}}{|S^{N-n-1}|\norm{u-\rho v}^{N-n}} \\
&= \frac{R(1-\rho^2)^{1-n/2}}{|S^{N-n-1}|A^{\frac{N-n}{2}}\norm{\omega - s\sigma}^{N-n}} \\
&= \left(\frac{(1-\rho^2)^{1-n/2}}{(1-r^2)A^{\frac{N-n}{2}}}\right) \left(\frac{R(1-r^2)}{|S^{N-n-1}|\norm{w - r x}^{N-n}}\right) \\
&= U_{\rho,N}(y,z)  \frac{R(1-r^2)}{|S^{N-n-1}|\norm{w - r x}^{N-n}}.
\end{align*}
The fact that $r \in [0,1)$ follows from Lemma~\ref{lem:A-bound}, which prove below in Appendix~\ref{sec:A-bound}.
\end{proof}

\subsection{Proof of Lemma~\ref{lem:mehler-limit}}
\begin{customlem}{\ref{lem:mehler-limit}}
$\lim_{N \to \infty} U_{\rho,N}(y,z) = U_{\rho}(y,z)$.
\end{customlem}
This lemma is stated without proof in \cite{CL90}.  We give a proof for completeness.
\begin{proof}
First, note that
\[
\lim_{N \to \infty} A(y,z) = \frac{1+\rho^2}{2} + \sqrt{\left(\frac{1+\rho^2}{2}\right)^2 - \rho^2} = \frac{1+\rho^2}{2} + \frac{1-\rho^2}{2} = 1,
\]
so
\[
\lim_{N \to \infty}r(y,z) = \rho.
\]
Therefore, it suffices to show that
\[
\lim_{N \to \infty} A(y,z)^{\frac{N-n}{2}} = \exp\left(\frac{(\rho^2(\norm{y}^2 + \norm{z}^2) -2\rho \la y,z \ra)}{2(1-\rho^2)}\right).
\]
An easy calculation shows that 
\[
\left(\frac{1+\rho^2-\frac{2\rho}{R^2} \la y,z \ra}{2}\right)^2-\rho^2\left(1-\frac{\norm{y}^2}{R^2}\right)\left(1-\frac{\norm{z}^2}{R^2}\right) = \left(\frac{1-\rho^2}{2} + \frac{\rho^2(\norm{y}^2 + \norm{z}^2) - \rho(1+\rho^2)\la y,z \ra + o(1)}{R^2(1-\rho^2)}\right)^2.
\]
Plugging this in to the definition of $A$, we get
\[
A(y,z) = 1 + \frac{\rho^2(\norm{y}^2 + \norm{z}^2) - 2\rho\la y,z \ra + o(1)}{R^2(1-\rho^2)}.
\]
Since $\frac{N-n}{2} = R^2/2 + o(R^2)$,
\[
\lim_{N \to \infty} A(y,z)^{\frac{N-n}{2}} = \exp\left(\frac{-(\rho^2(\norm{y}^2 + \norm{z}^2) -2\rho \la y,z \ra)}{2(1-\rho^2)}\right)
\]
as desired.
\end{proof}

\subsection{Proof of Lemma~\ref{lem:poisson-int}}
\begin{customlem}{\ref{lem:poisson-int}}
$\int_{S^{N-n-1}_R} \frac{R(1-r^2)}{|S^{N-n-1}|\norm{w - r x}^{N-n}} \,ds(x) = 1$.
\end{customlem}
To prove this, we need the following corollary of the Poisson Integral Formula (e.g., Theorem~3.43 of \cite{MP10}).
\begin{corollary}
\label{cor:poisson-int-formula}
For $0 \leq r < 1$,
\[
\int_{S^{N-1}} \frac{1-r^2}{\norm{u - rv}^{N}} d\omega(v) = 1.
\]
\end{corollary}
\begin{proof}[Proof of Lemma~\ref{lem:poisson-int}]
Using Corollary~\ref{cor:poisson-int-formula}, a simple change of variables shows that
\[
\int_{S_R^{N-n-1}} \frac{1-r^2}{\norm{w - rx}^{N-n}} ds(x)= \frac{|S^{N-n-1}|}{R}. \qedhere
\]
\end{proof}

\subsection{Proof of Lemma~\ref{lem:A-bound}} \label{sec:A-bound}
\begin{customlem}{\ref{lem:A-bound}}
$\left(1-\frac{\norm{y}^2}{R^2}\right)^{1/2} \left(1-\frac{\norm{z}^2}{R^2}\right)^{1/2} \leq A(y,z)$.
\end{customlem}
The proof is given in \cite{CL90}.  We include it for completeness.  
\begin{proof}
Assume that $\norm{y} < R$ and $\norm{z} < R$.  Otherwise, the claim is trivial.  Define $A'$ as follows:
\[
A' = \frac{1+\rho^2 - \frac{2\rho}{R^2}\norm{y}\norm{z}}{2} + \sqrt{\left(\frac{1+\rho^2-\frac{2\rho}{R^2}\norm{y}\norm{z}}{2}\right)^2 - \rho^2\left(1-\frac{\norm{y}^2}{R^2}\right)\left(1-\frac{\norm{z}^2}{R^2}\right)}.
\]
By Cauchy-Schwarz, we know that $A' \leq A$, so it suffices to show that
\[
\left(1-\frac{\norm{y}^2}{R^2}\right)^{1/2}\left(1-\frac{\norm{z}^2}{R^2}\right)^{1/2} \leq A'.
\]
Now define $\alpha = \left(1-\frac{\norm{y}^2}{R^2}\right)^{1/2}$, $\beta = \left(1-\frac{\norm{z}^2}{R^2}\right)^{1/2}$, and let $B = \frac{1+\rho^2-2\rho\sqrt{1-\alpha^2}\sqrt{1-\beta^2}}{2\alpha\beta}$.
Then
\[
B + \sqrt{B^2 - \rho^2} = \frac{1+\rho^2-2\rho\sqrt{1-\alpha^2}\sqrt{1-\beta^2}}{2\alpha\beta} + \frac{\sqrt{\left(\frac{1+\rho^2-2\rho\sqrt{1-\alpha^2}\sqrt{1-\beta^2}}{2}\right)^2 - \rho^2\alpha^2\beta^2}}{\alpha \beta} = \frac{A'}{\alpha\beta},
\]
so we will show that
\[
1 \leq B + \sqrt{B^2-\rho^2}.
\]
This statement, in turn, is implied by
\[
\frac{1+\rho^2}{2} \leq B.
\]
To prove this, observe that for any $\alpha, \beta$,
\[
(1-\alpha^2)(1-\beta^2) \leq (1-\alpha\beta)^2
\]
and for any $\rho$,
\[
2\rho \leq 1+\rho^2.
\]
Then
\[
2\rho\sqrt{1-\alpha^2}\sqrt{1-\beta^2} \leq (1+\rho^2)(1-\alpha\beta).
\]
Rearranging, we see that
\[
\frac{1+\rho^2}{2} \leq \frac{1+\rho^2-2\rho\sqrt{1-\alpha^2}\sqrt{1-\beta^2}}{2\alpha\beta} = B. \qedhere
\]
\end{proof}

\end{document}